\documentclass[12pt]{article}
\usepackage{geometry}
\usepackage{setspace}
\usepackage{amssymb,amsmath, amsthm, bbm, amsfonts, eurosym, graphicx, caption, color, setspace, sectsty, comment, pdflscape, array, hyperref, subcaption, soul, xcolor, tikz-network, natbib, enumitem, algorithm, algpseudocode, float}

\usepackage{mathtools}

\hypersetup{
  colorlinks,
  citecolor= black,
  linkcolor=black,
  urlcolor=black}

\makeatletter
\newcommand{\succprec}{\mathrel{\mathpalette\succ@prec{\succ\prec}}}
\newcommand{\succ@prec}[2]{\succ@@prec#1#2}
\newcommand{\succ@@prec}[3]{%
  \vcenter{\m@th\offinterlineskip
    \sbox\z@{$#1#3$}%
    \hbox{$#1#2$}\kern-0.4\ht\z@\box\z@
  }%
}
\makeatother

\usepackage[hang, flushmargin]{footmisc}

\newtheorem{example}{Example}

\newtheorem{remark}{Remark}

\newtheorem{proposition}{Proposition}

\setcitestyle{round}

\begin{document}

\title{Very Justified Envy\thanks{I thank Bobby Pakzad-Hurson for helpful comments.}
\author{Adam Hamdan\thanks{Department of Economics, Brown University, \texttt{\href{mailto:adam_hamdan@brown.edu}{adam\_hamdan@brown.edu}}}}}
\maketitle

\begin{abstract}
    \singlespacing
    This paper studies the intensity of priority violations in matching markets. An agent is said to have \emph{very justified envy} if she prefers another agent's match to her own and her priority advantage for the object exceeds $k$ ranks. I show that unless $k$ is prohibitively large, Pareto efficiency is incompatible with the elimination of very justified envy. A simple generalization of Deferred Acceptance and Immediate Acceptance is shown to eliminate very justified envy for any $k$. In simulations of random matching markets, I find that modest relaxations to justified envy can yield significant improvements in assignment ranks.
    

\end{abstract}


\newpage
\section{Introduction}

In a priority-based matching problem, a set of $n$ agents must be matched to a set of $m$ objects based on agents' preference rankings over objects, and objects' priority rankings over agents. In the context of school choice markets, these priorities may reflect a student's proximity to a school, her sibling's enrollment, or her performance on an entrance exam. In public housing allocation, priority may be given on the basis of income, terminal illness, duration of homelessness, or senior citizen status. Since the work of \citet{abdulkadirouglu2003school}, the standard approach to fairness in these problems has consisted in the elimination of \emph{justified envy}: no agent should prefer another agent's match while holding higher priority for that object.

By definition, justified envy does not capture the severity of priority violations. In practice, however, different cases of justified envy can vary substantially in the extent to which the envious agent's priority is violated. A student displaced at her neighborhood school by the very last applicant in that school's priority order may be more likely to appeal her assignment than one displaced by her immediate successor. Likewise, a public housing authority may tolerate a violation between two applicants of approximately similar need, but not one involving an applicant with far stronger claims to a specific housing unit. Treating all priority violations identically therefore discards information that matters to market participants and planners alike. 

This paper analyzes the intensity of priority violations in matching markets and its implications for market design. Specifically, I propose an intuitive generalization to the concept of justified envy. An agent is said to have \emph{very justified envy} if she prefers another agent's assignment to her own, and her priority advantage for that object strictly exceeds $k$ ranks. The threshold $k$ is meant to capture the degree of priority violation deemed ``acceptable'' in the market under consideration. If the threshold is set to zero, very justified envy reverts back to standard justified envy. At the other extreme, if the threshold is equal to $n - 1$, no priority violation can be classified as very justified. For thresholds between zero and $n - 1$, one opens the door to cases in which the priority difference between agents is significant enough to make one's envy justified, but not large enough to make it ``very'' justified. With this flexible framework in hand, several facts are established.

The first main finding is that tolerating violations of bounded intensity is generally not enough to escape the well-known tension between efficiency and fairness \citep{roth1982economics, balinski1999tale, abdulkadirouglu2003school}. Specifically, unless $k \ge n-2$, it is always possible to find a problem in which no matching is both Pareto efficient and free of very justified envy (\hyperref[prop1]{Proposition 1}). Since priority differences between agents are at most $n-1$ ranks, this implies that Pareto efficiency can be guaranteed only when nearly maximal priority violations are tolerated. 

Motivated by this fact, I consider a natural family of mechanisms, termed $k$-Threshold Acceptance ($k$-TA) mechanisms, which use the threshold $k$ to determine an appropriate matching. $k$-TA makes a straightforward modification to \cite{gale1962college}'s celebrated Deferred Acceptance (DA) algorithm: at each round, a proposing agent can displace a tentative match at a given object only if her priority advantage for that object strictly exceeds $k$ ranks. When $k = 0$, the procedure becomes equivalent to the original DA algorithm. When $k = n - 1$, no displacement is ever possible, making $k$-TA equivalent to the Immediate Acceptance (IA) or Boston mechanism. For any threshold $k$, the matching selected by $k$-TA is shown to be free of very justified envy (\hyperref[prop2]{Proposition 2}). 

I then highlight ways in which the boundaries of the $k$-TA family possess uniquely appealing properties. I find that $k$-TA is strategyproof if and only if it is equivalent to DA (\hyperref[prop3]{Proposition 3}) and is Pareto efficient if and only if it is equivalent to IA (\hyperref[prop4]{Proposition 4})---a familiar pattern among mechanism classes that generalize DA and IA \citep{chen2017chinese, ayoade2023school}. Furthermore, $k$-TA is constrained efficient among mechanisms that eliminate very justified envy if and only if it coincides with either DA or IA (\hyperref[prop5]{Proposition 5}). That is, unless $k = 0$ or $k = n - 1$, the matching selected by $k$-TA may be Pareto dominated by matchings which are free of very justified envy. Thus, while $k$-TA offers a range of intermediate mechanisms, moving away from DA and IA comes at the cost of manipulability and inefficiency.

Finally, I document the potential usefulness of the very justified envy paradigm by applying $k$-TA to simulated random matching markets, where preferences and priorities are drawn independently and uniformly at random. I find that, while the severity and volume of priority violations rises with the threshold $k$, so does the number of first-choice assignments, with most of the gains captured at early increments of $k$. Practically speaking, these findings suggest that even moderate relaxations to justified envy may deliver substantial welfare improvements over DA, with limited priority violations compared to IA.\footnote{These interpretations assume truthful preferences. See \hyperref[section4]{Section 4} for a discussion of this assumption.}

\paragraph{Related Literature.} This paper joins an extensive literature exploring alternatives to justified envy (or stability) in priority-based matching. This includes reasonable and consent-based stability \citep{kesten2004student, kesten2010school}, just assignments \citep{morrill2015making}, $\alpha$-equitability \citep{alcalde2017fair}, sticky stability \citep{afacan2017sticky}, partial stability \citep{dur2019school}, essential stability \citep{troyan2020essentially}, legality \citep{ehlers2020legal}, weak stability \citep{tang2021weak}, priority-neutrality \citep{reny2022efficient}, and justifiability \citep{ortega2026justifiable}. Of these, the closest work to mine is \cite{afacan2017sticky}, who allow justified envy to arise when an agent's desired object sits no more than $k$ preference-ranks above her current assignment. The elimination of very justified envy can therefore be seen as an object-side counterpart to their sticky-stability notion. Our respective mechanisms, however, are conceptually distinct. Since \cite{dur2019school} allow for any general correspondence of allowable priority violations, the elimination of very justified envy constitutes a special case of their partial stability framework, where a priority violation is allowable if and only if the envious agent's priority advantage is at most $k$ ranks. This also means that a constrained efficient and very-justified-envy-free matching can be obtained using the algorithm proposed in \cite{dur2019school} (see \hyperref[example3]{Example 3}). Finally, the current set of results relates to the recent analysis in \cite{ortega2026justifiable}: $k$-TA allows priority violations of at most $k$ priority-ranks, while their Just-Below-Cutoffs mechanism may violate an agent's priority by at most the number of \textit{unimprovable} agents---those agents whose assignments are weakly worse under any matching that Pareto dominates the DA matching. 

The mechanism family considered in this paper complements existing works generalizing DA and IA. \cite{chen2017chinese} study a parametric family of application-rejection mechanisms used for centralized college admission in China (which contains IA and DA at its extremes) and prove that its members can be ranked by their vulnerability to manipulation and their stability. More recently, \cite{ayoade2023school} consider priority-rank partitioned (PRP) rules, a family of student-proposing DA rules in which each school selects applicants through a choice function determined by a partition of its priority ranking and a partition of each student's preference ranking. PRP rules generalize several matching mechanisms in the literature, including those considered in \cite{chen2017chinese}, as well as DA, IA, and others. Because $k$-TA uses a window of $k$ ranks that slides with each agent rather than a fixed partition of the priority order, it is a PRP rule only at the extremes $k=0$ and $k=n-1$.

\paragraph{Outline.} The remainder of this paper is organized as follows. \hyperref[section2]{Section 2} presents a standard model of priority-based matching and describes the $k$-TA mechanism. \hyperref[section3]{Section 3} covers the main results on fairness, manipulability and efficiency.  \hyperref[section4]{Section 4} reports the simulation output. \hyperref[section5]{Section 5} concludes.

\section{Model}\label{section2}

A priority-based matching problem is a tuple $(I, S, q, P, \succ)$ consisting of a finite set of agents $I = \{i_1, \dots, i_n\}$ and a finite set of objects $S = \{s_1, \dots, s_m\}$, together with a capacity vector $q = (q_s)_{s \in S}$, where $q_s \in \{1, 2, \dots\}$ is the number of available copies (or ``seats'') of object $s$. The list $P = (P_i)_{i \in I}$ describes agents' preferences, where $P_i$ is agent $i$'s strict preferences over $S \cup \{\emptyset\}$, and $\emptyset$ stands for the option of being unassigned. Let $R_i$ denote the at-least-as-good-as relation associated with $P_i$, and say that object $s$ is \emph{acceptable} to $i$ if $s R_i \emptyset$. The list $\succ \, = (\succ_s)_{s \in S}$ describes objects' priorities, where $\succ_s$ is the complete and strict priority ranking of object $s$ over $I$. Unless stated otherwise, it is assumed that $\sum_{s \in S} q_s = n$, and $n,m\ge 4$.\footnote{The second assumption is needed for \hyperref[prop4]{Proposition 4} and \hyperref[prop5]{Proposition 5}. All other results hold for $n,m \ge 3$.}

We write $i \succ_s j$ to say that $i$ has strictly higher priority than $j$ at object $s$. It is convenient to represent priorities numerically: for each object $s \in S$, let $\mathrm{rk}_s : I \to \{1, \dots, n\}$ denote the rank function associated with $\succ_s$, defined so that $\mathrm{rk}_s(i) > \mathrm{rk}_s(j)$ if and only if $i \succ_s j$. That is, the agent with the highest priority at $s$ has rank $n$, and the agent with the lowest priority has rank $1$. For any two agents $i, j \in I$, I refer to the difference $\mathrm{rk}_s(i) - \mathrm{rk}_s(j)$ as $i$'s \emph{priority advantage} over $j$ at object $s$.

A matching is a function $\mu : I \to S \cup \{\emptyset\}$ such that $|\mu^{-1}(s)| \le q_s$ for each $s \in S$. A matching $\mu$ Pareto dominates another matching $\mu'$ if $\mu(i) P_i \mu'(i)$ for some $i \in I$ and $\mu(i) R_i \mu'(i)$ for all $i \in I$. A matching $\mu$ is \emph{Pareto efficient} if no other matching Pareto dominates it. Given a matching $\mu$, say that agent $i$ has \emph{justified envy} towards $s$ if $s P_i \mu(i)$ and either $i \succ_s j$ for some $j \in \mu^{-1}(s)$, or $|\mu^{-1}(s)| < q_s$.

To account for the severity of justified envy cases, the planner specifies a priority threshold $k \in \{0, 1, \dots, n - 1\}$. Say that agent $i$ has \emph{very justified envy} towards $s$ if $s P_i \mu(i)$ and either $\mathrm{rk}_s(i) - \mathrm{rk}_s(j) > k$ for some $j \in \mu^{-1}(s)$, or $|\mu^{-1}(s)| < q_s$; that is, $i$ prefers $s$ to her own assignment and either her priority advantage over some agent assigned to $s$ strictly exceeds the threshold $k$, or $s$ has an unfilled seat. A matching $\mu$ is \emph{free of very justified envy} if no agent has very justified envy under $\mu$. Standard justified envy corresponds to the case where $k = 0$.

A mechanism $\varphi$ selects a matching for each problem $(I, S, q, P, \succ)$, where $\varphi(P, \succ)$ denotes the matching $\mu$ selected by $\varphi$, and $\varphi_i(P, \succ)$ denotes the object assigned to agent $i$ under $\varphi$. A mechanism is Pareto efficient if, for any problem, it selects Pareto efficient matchings. A mechanism is free of justified envy if, for any problem, it selects matchings with no cases of justified envy. Given a threshold $k$, a mechanism $\varphi$ is free of very justified envy if, for any problem, it selects matchings that are free of very justified envy. Lastly, a mechanism $\varphi$ is \emph{strategyproof} if there is no problem featuring an agent $i$ and two preference rankings $P_i, \tilde{P_i}$ such that $\varphi_i(\tilde{P_i}, P_{-i}, \succ) P_i \varphi_i(P_i, P_{-i}, \succ)$.

I now introduce the class of $k$-Threshold Acceptance ($k$-TA) mechanisms, which proceed identically to DA except for the displacement rule used in each round. A formal description of $k$-TA (for a fixed $k$) is given below.

\begin{description}
\item[Round 1:] Every agent proposes to her favorite acceptable object. Every object $s$ tentatively assigns its seats to its proposers following its priority order $\succ_s$. Any remaining proposers are rejected.

\item In general, at

\item[Round $t\ge2$:] Every agent who is not tentatively matched proposes to her favorite acceptable object among those she has not previously proposed to. Every object $s$ first tentatively assigns any unoccupied seats to its proposers, one at a time following $\succ_s$. If $s$ is at full capacity and proposers remain, it considers its remaining proposers in decreasing priority order as follows. Let $i$ be the highest priority proposer, and let $j$ be the tentatively assigned agent with the lowest rank $\mathrm{rk}_s(j)$. If $\mathrm{rk}_s(i) > \mathrm{rk}_s(j) + k$, then $s$ rejects $j$, tentatively assigns the seat to $i$, and moves on to the next proposer. If instead $\mathrm{rk}_s(i) \le \mathrm{rk}_s(j) + k$, then $s$ rejects $i$ and all remaining proposers, since every remaining proposer $i'$ satisfies $\mathrm{rk}_s(i') < \mathrm{rk}_s(i)$ and every other tentatively assigned agent $j'$ satisfies $\mathrm{rk}_s(j') > \mathrm{rk}_s(j)$. 

\item The algorithm ends when no more proposals are made.
\end{description}

The next remark collects the fact that DA and IA are special instances of $k$-TA.

\begin{remark}
    $k$-TA nests the two most widely-used school choice mechanisms. For the smallest threshold, $k = 0$, $k$-TA coincides with the Deferred Acceptance mechanism \citep{gale1962college}, which has been used in cities like Denver and New York. For the largest threshold, $k = n - 1$, $k$-TA coincides with the Immediate Acceptance mechanism \citep{abdulkadirouglu2003school}, which has been used in cities like Cambridge and Providence.
\end{remark}

This section ends with a simple example illustrating how $k$-TA works for different threshold levels.

\begin{example}\label{example1}
Let $I = \{i_1, \dots, i_5\}$, $S = \{s_1,\dots, s_5\}$, and $q_s = 1$ for all $s \in S$. Preferences and priorities are as follows:
\begin{center}
\begin{tabular}{c c c c c | c c c c c}
    $P_{i_1}$ & $P_{i_2}$ & $P_{i_3}$ & $P_{i_4}$ & $P_{i_5}$ & $\succ_{s_1}$ & $\succ_{s_2}$ & $\succ_{s_3}$ & $\succ_{s_4}$ & $\succ_{s_5}$ \\
    \hline
    $s_1$ & $s_1$ & $s_2$ & $s_2$ & $s_3$ & $i_4$ & $i_2$ & $i_1$ & $i_3$ & $i_5$ \\
    $s_2$ & $s_3$ & $s_4$ & $s_1$ & $s_5$ & $i_5$ & $i_3$ & $i_5$ & $i_1$ & $i_1$ \\
    $s_3$ & $s_5$ & $s_3$ & $s_4$ & $s_2$ & $i_3$ & $i_4$ & $i_3$ & $i_5$ & $i_3$ \\
    $s_5$ & $s_2$ & $s_1$ & $s_3$ & $s_1$ & $i_1$ & $i_5$ & $i_4$ & $i_2$ & $i_4$ \\
    $s_4$ & $s_4$ & $s_5$ & $s_5$ & $s_4$ & $i_2$ & $i_1$ & $i_2$ & $i_4$ & $i_2$ \\
\end{tabular}
\end{center}
Consider applying $k$-TA with three different thresholds, $k \in \{0, 2, 4\}$. In case of multiple proposers in any given round, the agent who is tentatively matched at the end of that round is underlined. Applying $0$-TA (DA) to this example yields the following round-by-round process:
\begin{center}
\begin{tabular}{c | c c c c c}
    Round & $s_1$ & $s_2$ & $s_3$ & $s_4$ & $s_5$ \\
    \hline
    1 & $\underline{i_1}, \, i_2$ & $\underline{i_3}, \, i_4$ & $i_5$ & &  \\
    2 & $i_1, \, \underline{i_4}$ & $i_3$ & $i_2, \, \underline{i_5}$ & &  \\
    3 & $i_4$ & $i_1, \, \underline{i_3}$ & $i_5$ & & $i_2$ \\
    4 & $i_4$ & $i_3$ & $\underline{i_1}, \, i_5$ & & $i_2$ \\
    5 & $i_4$ & $i_3$ & $i_1$ & & $i_2, \, \underline{i_5}$ \\
    6 & $i_4$ & $\underline{i_2}, \, i_3$ & $i_1$ & & $i_5$ \\
    7 & $i_4$ & $i_2$ & $i_1$ & $i_3$ & $i_5$ \\
\end{tabular}
\end{center}
The process unfolds according to the standard DA rules: in each round, agents propose to their favorite object (among those they have not previously proposed to) and objects tentatively hold on to proposers in order of priority. The final matching is:
\[
\mu^{0\text{-}\mathrm{TA}}
=\begin{pmatrix} i_1 & i_2 & i_3 & i_4 & i_5\\
s_3 & s_2 & s_4 & s_1 & s_5\end{pmatrix}
\]

To see how increasing the threshold $k$ changes the outcome of $k$-TA, consider applying $2$-TA to the example above. This results in the following process:
\begin{center}
\begin{tabular}{c | c c c c c}
    Round & $s_1$ & $s_2$ & $s_3$ & $s_4$ & $s_5$ \\
    \hline
    1 & $\underline{i_1}, \, i_2$ & $\underline{i_3}, \, i_4$ & $i_5$ & &  \\
    2 & $i_1, \, \underline{i_4}$ & $i_3$ & $i_2, \, \underline{i_5}$ & &  \\
    3 & $i_4$ & $i_1, \, \underline{i_3}$ & $i_5$ & & $i_2$ \\
    4 & $i_4$ & $i_3$ & $i_1, \, \underline{i_5}$ & & $i_2$ \\
    5 & $i_4$ & $i_3$ & $i_5$ & & $\underline{i_1}, \, i_2$ \\
    6 & $i_4$ & $i_2, \, \underline{i_3}$ & $i_5$ & & $i_1$ \\
    7 & $i_4$ & $i_3$ & $i_5$ & $i_2$ & $i_1$ \\
\end{tabular}
\end{center}
The first round of proposals and tentative matches under $2$-TA proceeds identically as in $0$-TA (indeed, the first round proceeds identically under any $k$-TA mechanism). In round 2 onward, any agent attempting to match with an occupied object will need a sufficiently large priority advantage to displace a tentatively matched agent. The first place where this takes effect is in determining $i_1$'s outcome in round 4: despite having a higher priority than $i_5$ at $s_3$, $i_1$ is no longer able to displace this tentative match (as she did under DA) since her priority advantage $\mathrm{rk}_{s_3}(i_1) - \mathrm{rk}_{s_3}(i_5) = 5 - 4 = 1$ does not strictly exceed the threshold $k = 2$. The same reasoning explains why $i_2$ is unable to displace $i_3$ at $s_2$ in round $6$. The final matching is: 
\[
\mu^{2\text{-}\mathrm{TA}}
=\begin{pmatrix} i_1 & i_2 & i_3 & i_4 & i_5\\
s_5 & s_4 & s_2 & s_1 & s_3\end{pmatrix}
\]
Finally, applying $4$-TA (IA) to the example above results in the following process:
\begin{center}
\begin{tabular}{c | c c c c c}
    Round & $s_1$ & $s_2$ & $s_3$ & $s_4$ & $s_5$ \\
    \hline
    1 & $\underline{i_1}, \, i_2$ & $\underline{i_3}, \, i_4$ & $i_5$ & &  \\
    2 & $\underline{i_1}, \, i_4$ & $i_3$ & $i_2, \, \underline{i_5}$ & &  \\
    3 & $i_1$ & $i_3$ & $i_5$ & $i_4$ & $i_2$ \\
\end{tabular}
\end{center}
In each round, agents make proposals according to the same principles in $0$-TA and $2$-TA. Unlike both of these mechanisms, however, displacements are no longer possible, making all tentative matches final. This is seen as early as in round $2$: after $i_4$ fails to match with $s_2$ in round $1$, she is unable to match with her second choice, $s_1$, in round $2$, despite having the highest priority in $\succ_{s_1}$. The final matching is: 
\[
\mu^{4\text{-}\mathrm{TA}}
=\begin{pmatrix} i_1 & i_2 & i_3 & i_4 & i_5\\
s_1 & s_5 & s_2 & s_4 & s_3\end{pmatrix}
\]
\end{example}

\section{Results}\label{section3}

To begin the analysis, consider the following example by \cite{roth1982economics} which shows that efficiency and the elimination of justified envy may be at odds.

\begin{example}\label{example2}
Let $I = \{i_1, i_2, i_3\}$, $S = \{s_1, s_2, s_3\}$, and $q_s = 1$ for all $s \in S$. Preferences and priorities are as follows:
\begin{center}
\begin{tabular}{c c c | c c c}
    $P_{i_1}$ & $P_{i_2}$ & $P_{i_3}$ & $\succ_{s_1}$ & $\succ_{s_2}$ & $\succ_{s_3}$ \\
    \hline
    $s_2$ & $s_1$ & $s_1$ & $i_1$ & $i_2$ & $i_2$ \\
    $s_1$ & $s_2$ & $s_2$ & $i_3$ & $i_1$ & $i_1$ \\
    $s_3$ & $s_3$ & $s_3$ & $i_2$ & $i_3$ & $i_3$ \\
\end{tabular}
\end{center}
There is a unique matching $\mu$ which is free of justified envy
\[
\mu = \begin{pmatrix} i_1 & i_2 & i_3 \\ s_1 & s_2 & s_3 \end{pmatrix}
\]
However, $\mu$ is Pareto dominated by the following matching $\mu'$
\[
\mu' = \begin{pmatrix} i_1 & i_2 & i_3 \\ s_2 & s_1 & s_3 \end{pmatrix}
\]
\end{example}

\bigskip
Notice that in the example above, although the matching $\mu'$ is not free of justified envy, it is free of very justified envy for $k \ge 1$, since $i_3$ envies $i_2$'s assignment at $s_1$, but her priority advantage is $\mathrm{rk}_{s_1}(i_3) - \mathrm{rk}_{s_1}(i_2) = 1$. This suggests that there might exist a threshold $k$ that is ``small'' relative to the market size $n$ which suffices to restore Pareto efficiency at a bounded fairness cost. The first result shows that this is false: no threshold below $n-2$ guarantees, across all problems, the existence of a matching that is both Pareto efficient and free of very justified envy. Since no priority advantage can ever exceed $n-1$ ranks, the threshold needed for a general existence result is within a single rank of the maximal conceivable violation.

\begin{proposition}[\textbf{Impossibility}]\label{prop1}
Unless $k = n - 2$ or $k = n - 1$, there exists a problem for which no matching is both Pareto efficient and free of very justified envy.
\end{proposition}

\begin{proof}
    Fix $k < n - 2$. Let $q_s=1$ for all $s\in S$ and consider the preference profile
    \[
    P_{i_1}: s_1,\,s_2 \qquad P_{i_2}: s_1,\,s_2 \qquad P_{i_3}: s_2,\,s_1 \qquad P_{i_\ell}: s_\ell \ \ \forall \ell = 4,\dots, n
    \]
    where unlisted objects are unacceptable, together with priorities
\[
\succ_{s_1}:\ i_3,\,\dots,\,i_2,\,i_1, \qquad \succ_{s_2}:\ i_1,\,i_2,\,\dots,\,i_3,
\]
so that $\mathrm{rk}_{s_1}(i_3)=n$, $\mathrm{rk}_{s_1}(i_2)=2$, $\mathrm{rk}_{s_1}(i_1)=1$, and $\mathrm{rk}_{s_2}(i_1)=n$, $\mathrm{rk}_{s_2}(i_2)=n-1$, $\mathrm{rk}_{s_2}(i_3)=1$ and all other priorities $\succ_{s_\ell}$ are arbitrary for $\ell \ge 3$.

Since $s_\ell$ is the unique acceptable object for $i_\ell$ and $i_\ell$ the unique agent for whom $s_\ell$ is acceptable for $\ell\ge4$, then every Pareto efficient matching assigns $\mu(i_\ell)=s_\ell$. Since $i_1,i_2,i_3$ all find both $s_1,s_2$ acceptable and no other object acceptable, Pareto efficiency further requires that $s_1,s_2$ both be assigned within $\{i_1,i_2,i_3\}$, leaving one agent unmatched. Among the six ways to assign $s_1,s_2$, only four are Pareto efficient:
\[
\hat\mu = \begin{pmatrix} i_1 & i_2 & i_3 \\ s_1 & s_2 & \emptyset \end{pmatrix}, \quad
\check\mu = \begin{pmatrix} i_1 & i_2 & i_3 \\ s_1 & \emptyset & s_2 \end{pmatrix}, \quad
\tilde\mu = \begin{pmatrix} i_1 & i_2 & i_3 \\ s_2 & s_1 & \emptyset \end{pmatrix}, \quad
\bar\mu = \begin{pmatrix} i_1 & i_2 & i_3 \\ \emptyset & s_1 & s_2 \end{pmatrix}
\]
Each of these assignments features a case of very justified envy. Under $\hat\mu$, $i_3$ envies $i_1$'s match at $s_1$ with an advantage of $n-1$ ranks. Under $\check\mu$, $i_2$ envies $i_3$'s match at $s_2$ with an advantage of $n-2$ ranks. Under $\tilde\mu$, $i_3$ envies $i_2$'s match at $s_1$ with an advantage of $n-2$ ranks. Under $\bar\mu$, $i_1$ envies $i_3$'s seat at $s_2$ with an advantage of $n-1$ ranks. Since $k < n-2$, every one of these advantages exceeds $k$. Therefore, all Pareto efficient matchings in this problem exhibit very justified envy.

Now let $k = n - 2$, and fix a problem $(I, S, q, P, \succ)$. We establish existence of a matching that is Pareto efficient and free of very justified envy by constructing a serial dictatorship (SD) mechanism that selects one. Recall that in SD, agents take turns in a pre-determined serial order, $\pi$, matching with their favorite acceptable object that still has remaining capacity. Since the outcome of SD is Pareto efficient for any $\pi$ \citep{abdulkadirouglu1998random}, it suffices to exhibit an ordering $\pi^*$ whose outcome is free of very justified envy for $k = n - 2$. Consider the following $\pi^*$: at each step $t \ge 1$, let the dictator $\pi^*(t)$ be one whose most-preferred acceptable object among those with remaining capacity ranks her the highest among all currently unmatched agents; once no unmatched agent finds an
available object acceptable, the remaining agents are ordered arbitrarily. Let $s(t)$ denote the object selected by the $t$-th dictator $\pi^*(t)$, and let $\mathrm{rk}_{s(t)}(\pi^*(t))$ denote the rank of this $t$-th dictator at object $s(t)$. Starting at step $1$, suppose $\mathrm{rk}_{s(1)}(\pi^*(1)) = 1$. By construction, this means that every agent sits at the bottom of her favorite object's priority ranking, which only occurs when every agent has a different favorite object, in which case, any $\pi$ results in zero cases of justified envy from this step onward. Now suppose that $\mathrm{rk}_{s(1)}(\pi^*(1)) \ge 2$. Notice that in this case, regardless of the realized preference profile, no future agent $\pi^*(t)$ (with $t > 1$) can feel very justified envy towards $s(1)$: even the highest-priority agent at $s(1)$ has a priority advantage of at most $n - 2$ ranks over $\pi^*(1)$, which is insufficient for very justified envy given $k = n - 2$. Thus, for any rank $\mathrm{rk}_{s(1)}(\pi^*(1))$, no agent feels very justified envy towards $s(1)$. Repeating this argument shows that at any step $t \ge 1$, no agent feels very justified envy towards earlier assignments $s(1), ..., s(t-1)$. Since agents can only experience justified envy towards earlier assignments in any SD mechanism, the matching associated with $\pi^*$ is both Pareto efficient and free of very justified envy. 

Finally, let $k = n - 1$. Since no priority advantage can exceed $n - 1$ ranks and since SD is non-wasteful, using SD with \emph{any} serial order yields an outcome that is both Pareto efficient and free of very justified envy.
\end{proof}

An implication of \hyperref[prop1]{Proposition 1} is that it is generally impossible to design a mechanism that is both Pareto efficient and free of very justified envy. If one insists on satisfying the latter of these objectives, the next result states that $k$-TA is a suitable candidate.

\begin{proposition}[\textbf{Elimination of very justified envy}]\label{prop2}
    $k$-TA is free of very justified envy given any threshold $k \in \{0, \dots, n-1\}$.
\end{proposition}

\begin{proof}
    The proof closely follows that of Theorem 1 in \cite{gale1962college}. Fix a problem $(I,S,q, P,\succ)$ and a threshold $k$, and let $\mu$ denote the matching selected by $k$-TA. Suppose $s P_i \mu(i)$ for some agent $i$ and object $s$. Since $k$-TA is non-wasteful, it cannot be that $| \mu^{-1}(s) | < q_s$. Therefore, it suffices to show that $\mathrm{rk}_s(i) - \mathrm{rk}_s(j) \le k$ for every $j$ with $\mu(j) = s$, so that $i$'s envy is not very justified.
    
    Since $i$ prefers $s$ to her match under $\mu$, she must have proposed to $s$ at some step and was eventually rejected. Note that rejections occur only once $s$ is at full capacity, and from that point on the set of $s$'s tentative matches changes only through displacements, each of which removes the incumbent with the lowest rank and replaces her with a proposer whose rank exceeds it by more than $k$. In particular, the lowest rank among $s$'s tentative matches never decreases once $s$ reaches capacity.
    
    Suppose first that $i$ is rejected upon proposing, and let $\underline{j}$ be the lowest-ranked tentative match at the end of this round. Then $\mathrm{rk}_s(i) \le \mathrm{rk}_s(\underline{j}) + k$: either every seat was assigned to a proposer ranked above $i$, or $i$ (or a proposer ranked above her) failed the threshold comparison against $\underline{j}$. Since the lowest rank at $s$ never decreases thereafter, every final match $j$ satisfies $\mathrm{rk}_s(j) \ge \mathrm{rk}_s(\underline{j}) \ge \mathrm{rk}_s(i) - k$.
    
    Suppose instead that $i$ is tentatively matched to $s$ and displaced at a later step. At that moment $i$ is the lowest-ranked incumbent, so every other incumbent, as well as the displacing proposer, has rank exceeding $\mathrm{rk}_s(i)$; and since displacements only ever replace the lowest-ranked incumbent with a higher-ranked agent, every final match $j$ satisfies $\mathrm{rk}_s(j) > \mathrm{rk}_s(i)$. In both cases, $\mathrm{rk}_s(i) \le \mathrm{rk}_s(j) + k$ for all $j$ with $\mu(j) = s$.
\end{proof}

The previous result establishes that every $k$-TA mechanism is fair from the perspective of very justified envy. The next result shows that only a single member of the $k$-TA family has the additional property of being strategyproof, and that is DA.

\begin{proposition}[\textbf{Strategyproofness}]\label{prop3}
    $k$-TA is strategyproof if and only if $k = 0$.    
\end{proposition}

\begin{proof}
    The ``if'' direction follows from the fact that DA is strategyproof \citep{dubins1981machiavelli, roth1982economics}. The ``only if'' direction proceeds as follows. Fix a threshold $k \ge 1$. Let agent $i$'s true preferences be $s_1 P_i s_2$, with all other objects unacceptable. Let $I_1 \subset I$ consist of $q_{s_1}$ agents whose sole acceptable object is $s_1$, and $I_2 = I_2' \cup \{j\} \subset I$ consist of $q_{s_2}$ agents whose sole acceptable object is $s_2$. Let priorities be such that every agent in $I_1$ outranks $i$ at $s_1$, every agent in $I_2'$ outranks $i$ at $s_2$, every agent in $I \setminus \left( I_2 \cup \ \{i\} \right)$ has lower priority than $i$ at $s_2$, and $i$ outranks $j$ at $s_2$ by exactly $k$ ranks, so that $\mathrm{rk}_{s_2}(i) = \mathrm{rk}_{s_2}(j) + k$. 

    Under truthful reporting, $i$ and all agents in $I_1$ propose to $s_1$ in round 1. Since $I_1$ exactly fills $s_1$'s $q_{s_1}$ seats with agents who outrank $i$, $i$ is rejected. Simultaneously, the set $I_2$ exactly fills $s_2$. In round 2, $i$ proposes to $s_2$, now at capacity, and is compared against the lowest-ranked incumbent, $j$. Since $\mathrm{rk}_{s_2}(i) =  \mathrm{rk}_{s_2}(j) + k \not> \mathrm{rk}_{s_2}(j) + k$, $i$ fails to displace $j$ and, having exhausted her preferences, remains unmatched.

    Now suppose that $i$ misrepresents her preferences and reports $s_2 \tilde{P_i} s_1$ instead. In round 1, $i$ proposes to $s_2$ alongside $I_2$. Since $i$ outranks $j$, she obtains a seat at $s_2$, while $j$ is rejected. Moreover, since $i$ is among the $q_{s_2}$ highest-priority agents at $s_2$, and all higher-priority agents at $s_2$ have also secured a seat at $s_2$, $i$ cannot be displaced by a proposer in future rounds. Since $i$'s final match is $s_2$, a strict improvement over being unmatched, truthful revelation is not a dominant strategy for $i$ when $k \ge 1$. 
    
    Note that the construction above requires $q_{s_2} + k \le n$ to hold. For $k > n - q_{s_2}$, one can repeat the argument with $q_{s_1} = q_{s_2} = 1$ and $I_2' = \emptyset$, which holds for every $k \le n - 1$.
\end{proof}

A well-known weakness of DA is that it can produce highly inefficient outcomes \citep{abdulkadirouglu2003school, kesten2010school}. It is therefore useful to ask which members of the $k$-TA family are guaranteed to select Pareto efficient matchings. The next result shows that this holds if and only if $k$-TA coincides with IA.

\begin{proposition}[\textbf{Efficiency}]\label{prop4}
    $k$-TA is Pareto efficient if and only if $k = n - 1$.
\end{proposition}

\begin{proof}
    The ``if'' direction follows from the fact that IA is Pareto efficient. For the ``only if'' direction, fix $k < n - 1$. We construct a problem in which $k$-TA selects a matching that is not Pareto efficient. It suffices to treat the boundary case $k = n - 2$. For $k < n-2$, one may invoke \hyperref[prop1]{Proposition 1} to find a problem in which no matching is both Pareto efficient and free of very justified envy; since $k$-TA is free of very justified envy by \hyperref[prop2]{Proposition 2}, its matching in such a problem is necessarily inefficient.

    Let $q_s = 1$ for all $s \in S$, and consider the following preference profile
    \[
        P_{i_1}: s_1
        \qquad
        P_{i_2}: s_2
        \qquad
        P_{i_3}: s_1, \, s_2, \, s_3
        \qquad
        P_{i_4}: s_3,\, s_2 \qquad P_{i_\ell}: s_\ell, 
    \ \forall\ell = 5, \dots, n
    \]
    together with priorities satisfying
    \[
        \succ_{s_1}: i_1,\, i_3,\, \dots
        \qquad
        \succ_{s_2}: i_4,\, \dots,\, i_3, \, i_2
        \qquad
        \succ_{s_3}: i_3,\, \dots,\, i_4
    \]
    and the remaining priorities $\succ_{s_\ell}$ for $\ell \ge 4$ are arbitrary. Notice that in the first round of $k$-TA, each agent $i_\ell$ with $\ell \ge 5$ is immediately and permanently matched to $s_\ell$, so one may restrict attention to $\{i_1, i_2, i_3, i_4\}$. The run of $k$-TA proceeds as follows. In round 1, $i_1$ and $i_3$ propose to $s_1$. Since $i_1 \succ_{s_1} i_3$, $i_3$ is rejected. Meanwhile, $i_2$ and $i_4$ tentatively match with $s_2$ and $s_3$, respectively. In round 2, $i_3$ proposes to $s_2$, but her priority advantage is not large enough, since $\mathrm{rk}_{s_2}(i_3) - \mathrm{rk}_{s_2}(i_2) = 2 - 1 = 1 \le k$, so $i_3$ is rejected. In round 3, $i_3$ proposes to $s_3$, where she is top-ranked. The incumbent, $i_4$ has priority rank $1$, which implies $\mathrm{rk}_{s_3}(i_3) - \mathrm{rk}_{s_3}(i_4) = n - 1 > k$ and so $i_3$ displaces $i_4$. In the fourth and final round, $i_4$ proposes to $s_2$, where she is top-ranked. The incumbent, $i_2$, has the lowest rank at $s_2$, which means $\mathrm{rk}_{s_2}(i_4) - \mathrm{rk}_{s_2}(i_2) = n - 1 > k$, and $i_4$ displaces $i_2$. No more proposals are made, and the final match from $k$-TA is
    \[
    \mu = \begin{pmatrix} i_1 & i_2 & i_3 & i_4 & i_5 & \dots & i_{n} \\ s_1 & \emptyset & s_3 & s_2 & s_5 & \dots & s_n \end{pmatrix}
    \]
    Since $i_3$ and $i_4$ both prefer to switch objects, $\mu$ is Pareto dominated by the matching
    \[
    \mu' = \begin{pmatrix} i_1 & i_2 & i_3 & i_4 & i_5 & \dots & i_{n} \\ s_1 & \emptyset & s_2 & s_3 & s_5 & \dots & s_n \end{pmatrix}
    \]
\end{proof}

Thus, with the exception of its boundary $k = n - 1$, $k$-TA may produce inefficient outcomes. In fact, a closer look at the proof of \hyperref[prop4]{Proposition 4} reveals a more subtle issue: for some problem instances, the matching selected by $k$-TA may be Pareto dominated by matchings which are themselves free of very justified envy. Say that a mechanism $\varphi$ is \emph{constrained efficient} in the class of mechanisms $\Phi$ if for any problem, the matching selected by $\varphi$ is not Pareto dominated by the matching selected by another mechanism $\psi \in \Phi$. The next result states that $k$-TA is constrained efficient among very-justified-envy-free mechanisms if and only if it coincides with DA, or IA. 

\begin{proposition}[\textbf{Constrained efficiency}]\label{prop5}
    $k$-TA is constrained efficient in the class of mechanisms that are free of very justified envy if and only if $k = 0$ or $k = n - 1$. 
\end{proposition}

\begin{proof}
    The ``if'' direction follows from the fact that IA is Pareto efficient (hence, constrained efficient) while DA weakly Pareto dominates all other mechanisms that eliminate justified envy \citep{gale1962college}.

    For the ``only if'' direction, I refer to the construction used in the proof of \hyperref[prop4]{Proposition 4}. Notice that for any threshold $k \in \{1,\dots, n-2\}$, $k$-TA always selects the same matching $\mu$: the two displacements involved a priority advantage of $n - 1$, and the sole rejection was due to a priority advantage of $1$. However, $\mu$ is Pareto dominated by $\mu'$, which is free of justified envy, and therefore, free of very justified envy for any $k$. Thus, $k$-TA is not constrained efficient for this range of $k$: simply take the mechanism that selects $\mu'$ at this problem and agrees with $k$-TA everywhere else.
\end{proof}

One way towards constrained efficiency is to consider improvement cycles, where one starts with a baseline matching and executes Pareto improving cycles where possible. This is precisely the idea explored in \cite{dur2019school}. Their work shows that to satisfy any desired notion of partial stability (such as the elimination of very justified envy), it suffices to (i) start with a partially stable baseline matching, and (ii) restrict attention to Pareto improving cycles that only result in ``allowable'' priority violations. Both of these principles are infused into the Top Priority (TP) algorithm \citep{dur2019school}, which (i) uses the outcome of DA as its baseline matching, and (ii) in case of multiple allowable improvement cycles, selects the cycle featuring the envious agent with the highest priority at her desired school.\footnote{For a formal description of the TP algorithm, please refer to \cite{dur2019school}.} The next example applies TP to find a constrained efficient and very-justified-envy-free matching.

\begin{example}\label{example3}
Let $I=\{i_1,\dots,i_5\}$, $S=\{s_1,\dots,s_5\}$, and $q_s=1$ for all $s\in S$. Fix the very justified envy threshold to be $k=1$. Preferences and priorities are as follows:
\begin{center}
\begin{tabular}{c c c c c | c c c c c}
    $P_{i_1}$ & $P_{i_2}$ & $P_{i_3}$ & $P_{i_4}$ & $P_{i_5}$ & $\succ_{s_1}$ & $\succ_{s_2}$ & $\succ_{s_3}$ & $\succ_{s_4}$ & $\succ_{s_5}$ \\
    \hline
    $s_4$ & $s_4$ & $s_1$ & $s_4$ & $s_3$ & $i_3$ & $i_4$ & $i_2$ & $i_5$ & $i_2$ \\
    $s_2$ & $s_1$ & $s_4$ & $s_2$ & $s_4$ & $i_2$ & $i_1$ & $i_1$ & $i_3$ & $i_5$ \\
    $s_5$      & $s_3$ &       &       &       & $i_4$ & $i_5$ & $i_5$ & $i_4$ & $i_1$ \\
          &       &       &       &       & $i_1$ & $i_2$ & $i_3$ & $i_2$ & $i_3$ \\
          &       &       &       &       & $i_5$ & $i_3$ & $i_4$ & $i_1$ & $i_4$ \\
\end{tabular}
\end{center}
Applying $0$-TA (DA) yields the following matching: 
\[
\mu^{\mathrm{DA}}=\begin{pmatrix} i_1 & i_2 & i_3 & i_4 & i_5 \\
s_5 & s_3 & s_1 & s_2 & s_4
\end{pmatrix}
\]
Applying $1$-TA yields the following matching: 
\[
\mu^{1\text{-}\mathrm{TA}}=\begin{pmatrix} i_1 & i_2 & i_3 & i_4 & i_5\\
s_2 & s_3 & s_1 & \emptyset & s_4\end{pmatrix}
\]
Applying \cite{dur2019school}'s TP algorithm yields the following matching: 
\[
\mu^{\mathrm{TP}}=\begin{pmatrix} i_1 & i_2 & i_3 & i_4 & i_5\\
s_5 & s_4 & s_1 & s_2 & s_3\end{pmatrix}
\]
The matching $\mu^{\mathrm{TP}}$ is obtained as follows. Let $\mu^{\mathrm{DA}}$ be the baseline matching. The only Pareto improving cycle in $\mu^{\mathrm{DA}}$ is $i_2 \to s_4$ and $i_5 \to s_3$. Note that executing this cycle creates justified envy, but not very justified envy: while $i_4$ would be justifiably envious of $i_2$, her priority advantage is $\mathrm{rk}_{s_4}(i_4) - \mathrm{rk}_{s_4}(i_2) = 1 \leq k$. Since this cycle only generates allowable priority violations, TP executes this cycle to obtain its final matching, $\mu^{\mathrm{TP}}$. Note that $\mu^{\mathrm{TP}}$ is free of very justified envy and Pareto efficient (thus, constrained efficient), while
$\mu^{\mathrm{DA}}$ and $\mu^{1\text{-}\mathrm{TA}}$ are both inefficient. 
\end{example}

\section{Simulations}\label{section4}

This section presents simulation evidence of the performance of $k$-TA in random matching markets. Each simulated market consists of $100$ agents and $100$ unit-capacity objects, with each agent's preferences and each object's priorities drawn independently and uniformly at random. Outcomes are averaged over $1{,}000$ independent draws, and are reported for five different threshold levels: $0$-TA (DA), $25$-TA, $50$-TA, $75$-TA, and $99$-TA (IA). 

Before proceeding, two comments are in order. First, in interpreting the output of these simulations, preferences are assumed to be reported truthfully, as is often done when considering mechanisms that are not strategyproof (e.g. \cite{ortega2023cost}). Indeed, understanding how any mechanism performs under truth-telling seems particularly relevant in light of recent experimental evidence suggesting that a large proportion of participants reveal their preferences truthfully irrespective of strategyproofness, with truth-telling rates sometimes higher under non-strategyproof mechanisms than under strategyproof ones (e.g. \cite{pais2008school, cerrone2024school}). Therefore, this exercise focuses solely on fairness and assignment ranks and abstracts from the issue of manipulability documented in \hyperref[prop3]{Proposition 3}.

Second, the primary purpose of these simulations is to highlight how relaxations to justified envy (through increases in $k$) can lead to improvements in assignment ranks under different members of $k$-TA. As established in \hyperref[prop5]{Proposition 5}, however, $k$-TA may sometimes leave some Pareto improvements on the table. In this sense, k-TA is purely meant as an operationally simple implementation of the elimination of very justified envy, rather than an optimal one. Accordingly, the findings below underestimate the assignment rank gains possible from allowing small priority violations. 

\paragraph{Intensity of justified envy.} \hyperref[fig1]{Figure 1} reports the number of justified envy cases grouped by the envious agent's priority advantage at her desired object. In line with \hyperref[prop2]{Proposition 2}, each $k$-TA mechanism admits justified envy cases with priority differences of at most $k$ ranks: $0$-TA (DA) produces no justified envy; $25$-TA produces justified envy within the $1$--$25$ rank band (with an average priority difference of $11$ ranks); $99$-TA (IA) features justified envy in the $1$--$99$ band (with an average priority difference of $31$ ranks). The number of cases also grows with $k$, with an average of $0, 47, 83, 108$ and $112$ cases under $0$-TA, $25$-TA, $50$-TA, $75$-TA and $99$-TA, respectively. 

\begin{figure}[htbp]
    \centering
    \includegraphics[width=0.9\linewidth]{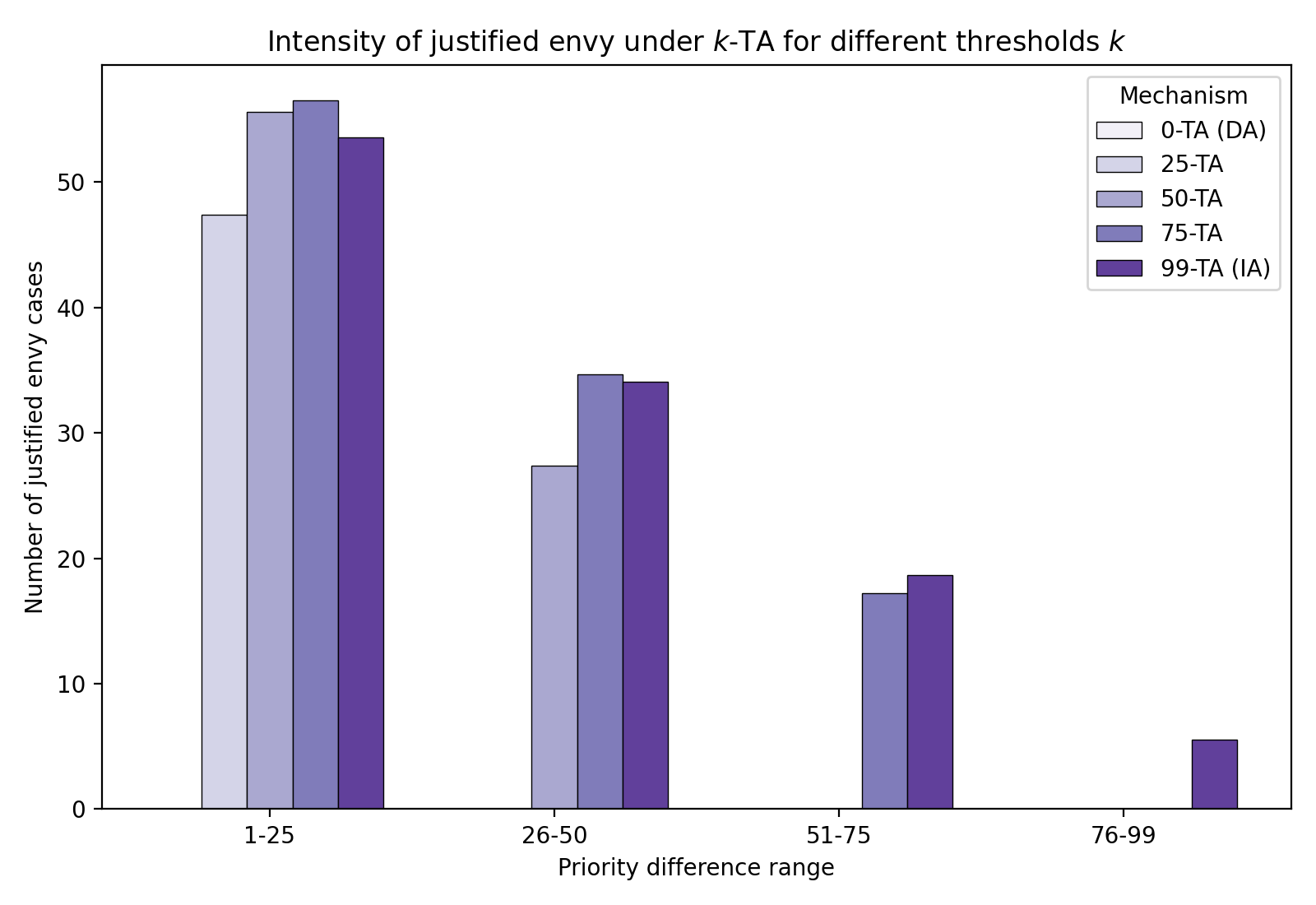}
    \caption{Intensity of justified envy under $k$-TA for $k \in \{0, 25, 50, 75, 99\}$. Results are averaged over $1{,}000$ independently simulated markets, each consisting of $100$ agents and $100$ unit-capacity objects, with preferences and priorities drawn independently and uniformly at random. Each case of justified envy is categorized by the envious agent's priority advantage over the agent she envies, binned into ranges $1$--$25$, $26$--$50$, $51$--$75$, and $76$--$99$.}
    \label{fig1}
\end{figure}

\paragraph{Assignment ranks.} \hyperref[fig2]{Figure 2} reports the distribution of assignment ranks (where an assignment of rank $r$ means the agent was assigned her $r$-th most preferred object). Raising the value of $k$ shifts the entire distribution toward lower-ranked (i.e. better) assignments. Consistent with prior work, the number of first-choice assignments is three times higher under IA ($63$ assignments) than under DA ($21$ assignments), while the number of poorly-ranked assignments (fifth-choice or worse) falls from about $40$ under DA to $16$ under IA. Interestingly, most improvements in the number of first-choice assignments are concentrated at lower values of $k$. For example, moving from $0$-TA (DA) to $25$-TA increases the number of first choice assignments from $21$ to $38$, bridging 40\% of the gap in first-choice assignments between DA and IA. By contrast, going from $75$-TA to $99$-TA (IA) only leads to $3$ additional first-choice assignments. 

\begin{figure}[htbp]
    \centering
    \includegraphics[width=0.9\linewidth]{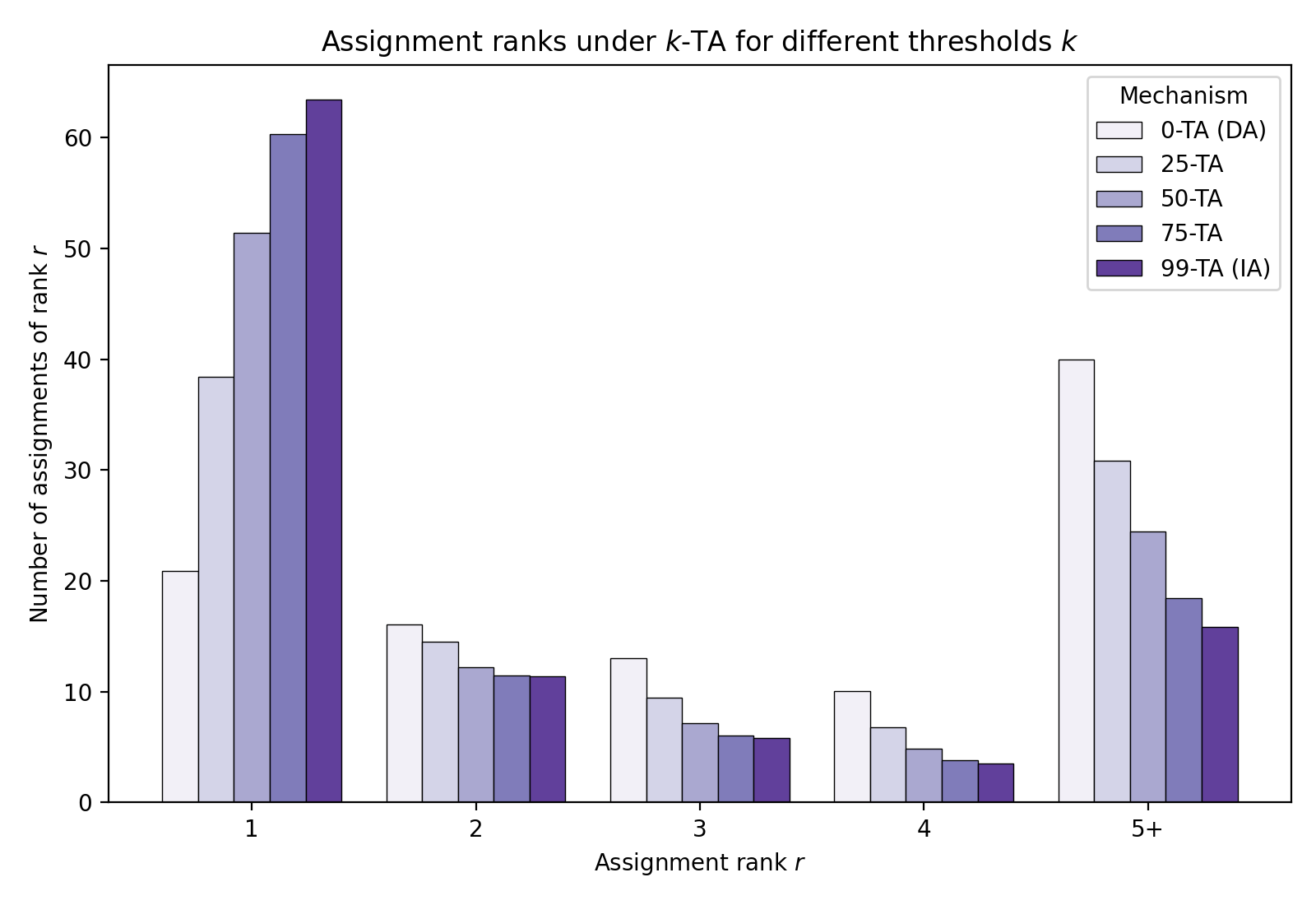}
    \caption{Distribution of assignment ranks under $k$-TA for $k \in \{0, 25, 50, 75, 99\}$. Results are averaged over $1{,}000$ independently simulated markets, each consisting of $100$ agents and $100$ unit-capacity objects, with preferences and priorities drawn independently and uniformly at random. An agent has assignment rank $r$ if the agent is matched to her $r$-th most preferred object.}
    \label{fig2}
\end{figure}

\paragraph{Tradeoff.} Taken together, the two figures display an important tradeoff in the choice of threshold $k$. A larger threshold improves the rank distribution of agents' assignments but permits both more numerous and more intense priority violations. On both margins, the early increments of $k$ appear to be relatively attractive: they deliver the bulk of the first-choice gains while keeping every violation within a narrow band of low-intensity envy. Later increments are comparatively costly, buying few additional first-choice assignments while opening the door to more severe priority violations. Overall, these findings suggest that allowing ``moderately'' justified envy can deliver substantial rank improvements with limited downsides on fairness.   

\section{Conclusion}\label{section5}

This paper analyzes the intensity of priority violations in matching markets such as school choice. I generalize the concept of justified envy to that of \emph{very justified envy}, which requires that an envious agent's priority be sufficiently larger than that of a currently matched one. While allowing bounded priority violations is usually not enough to reconcile Pareto efficiency with fairness, a simple generalization of Deferred Acceptance and Immediate Acceptance is shown to always select a matching which is free of very justified envy. Several properties of the proposed mechanism are theoretically established, and its empirical performance is assessed in simulated random matching markets. 

The current framework also raises open questions. In particular, a growing line of research is exploring how much justified envy emerges under standard efficient mechanisms like Top Trading Cycles (TTC), Serial Dictatorship (SD), and the Efficiency Adjusted Deferred Acceptance (EADA) \citep{abdulkadiroglu2020efficiency, abdulkadiroglu2020efficient, dougan2022robust, afacan2025improving,
kwon2026justified, hamdan2026making}, yet little is known about the severity of justified envy across these mechanisms. Following \cite{abdulkadiroglu2020efficiency} and \cite{dougan2022robust}, is TTC very-justified-envy-minimal among efficient and strategyproof mechanisms in one-to-one matching? Similarly, can one extend the findings in \cite{kwon2026justified} to show that EADA minimizes very justified envy among efficient mechanisms? These questions are left for future work.

\newpage
\singlespacing
\bibliographystyle{plainnat}
\bibliography{bib}

\end{document}